\documentclass[twocolumn,twoside]{IEEEtran}
\usepackage{mathpple}
\usepackage{times}

\usepackage{amsmath}  
\usepackage{amssymb}  
\usepackage{mathrsfs} 

\usepackage{theorem}  
\usepackage{cite}     
\usepackage{comment}  

\usepackage{upref}
\usepackage{amsfonts}

\usepackage{verbatim}

\usepackage[dvipsnames,usenames]{color}

\newcommand{\code}{\ttfamily\bfseries}

\newcommand{\dfn}{\bfseries\itshape}

\newcommand{\be}[1]{\begin{equation}\label{#1}}
\newcommand{\ee}{\end{equation}}
\newcommand{\eq}[1]{(\ref{#1})}
\newcommand{\Eq}[1]{(\ref{#1})}

\newcommand{\bc}{\begin{center}}
\newcommand{\ec}{\end{center}}

\newcommand{\Proof}{{{\it ~\,Proof. }}}

\def\qed{\hskip 3pt \hbox{\vrule width4pt depth2pt height6pt}}

\newcommand{\ceil}[1]{\lceil{#1}\rceil}

\newcommand{\cA}{{\cal A}}

\newcommand{\cD}{{\cal D}}
\newcommand{\cE}{{\cal E}}

\newcommand{\cG}{{\cal G}}

\newcommand{\cT}{{\cal T}}

\newcommand{\cV}{{\cal V}}

\newcommand{\bfu}{{\boldsymbol u}}
\newcommand{\bfv}{{\boldsymbol v}}

\newcommand{\bfx}{{\boldsymbol x}}
\newcommand{\bfy}{{\boldsymbol y}}

\renewcommand{\le}{\leqslant}
\renewcommand{\leq}{\leqslant}
\renewcommand{\ge}{\geqslant}
\renewcommand{\geq}{\geqslant}

\newcommand{\trunc}[1]{\left\lfloor {#1} \right\rfloor}

\newcommand{\Ftk}{\smash{\{0,1\}^k}}

\newcommand{\C}{\mathbb{C}}

\DeclareMathOperator{\wt}{wt}
\newcommand{\E}{\sf E}

\newcommand{\Tref}[1]{Theo\-rem\,\ref{#1}}

\newcommand{\Cref}[1]{Co\-rol\-la\-ry\,\ref{#1}}

\theoremstyle{plain} \theorembodyfont{\normalfont\slshape}

\newtheorem{thm}{Theorem$\!$}
\newenvironment{theorem}{\begin{thm}\hspace*{-1ex}{\bf.}}{\end{thm}}

\newtheorem{prop}[thm]{Proposition$\!$}

\newtheorem{lem}[thm]{Lemma$\!$}
\newenvironment{lemma}{\begin{lem}\hspace*{-1ex}{\bf.}}{\end{lem}}

\newtheorem{cor}[thm]{Corollary$\!$}

\newtheorem{defi}{Definition.}
\newenvironment{definition}{\begin{defi}\hspace*{-1.5ex}}{\end{defi}}

\theorembodyfont{\normalfont}

\newtheorem{exam}{Example$\!$}
\newenvironment{example}{\begin{exam}\hspace*{-1ex}{\bf .}}{\end{exam}}

\newtheorem{remrk}{Remark$\!$}

\definecolor{Codecolor}{named}{White}  
\def\screenrule{\hbox{\vrule width 0.85\columnwidth depth0pt height0pt}}

\newcommand{\codemode}[1]{%
\begin{center}%
\fcolorbox{black}{Codecolor}{%
\begin{ttfamily}\begin{bfseries}%
\shortstack[l]{%
\screenrule\\[0.00ex]%
#1\\[0.0ex]$\,$%
}
\end{bfseries}\end{ttfamily}%
}
\end{center}
}

\newcommand{\Copen}{\mbox{\{\kern-5.50pt\{}}
\newcommand{\Cclose}{\mbox{\}\kern-5.50pt\}}}
\newcommand{\Cslash}{\mbox{$\backslash\kern-6.02pt\backslash$}}

\begin{document}

\title{Rewriting Codes for Flash Memories}

\author{\large Eitan Yaakobi, Hessam Mahdavifar, Paul H. Siegel, Alexander Vardy, Jack K. Wolf,
\thanks{E.\ Yaakobi, H.\ Mahdavifar, P.H.\ Siegel, A.\ Vardy, and J.K.\ Wolf are with the Department of Electrical and Computer Engineering, University of California at San Diego, La Jolla, CA 92093, U.S.A. (e-mail: \{eyaakobi, hessam, psiegel, avardy, jwolf\}@ucsd.edu).}
}
\maketitle

\bibliographystyle{IEEEtranS}
\begin{abstract}
Flash memory is a non-volatile computer memory comprising blocks of cells, wherein each cell can take on~$q$~dif\-ferent values or \emph{levels}. While increasing the cell level is easy,~reducing the level of a cell can be accomplished only by erasing an entire block. Since block erasures are highly undesirable, coding schemes --- known as \emph{floating codes} (or \emph{flash codes}) and \emph{buffer codes} --- have~been~de\-signed in order to maximize the number of times that information stored in a flash memory can be written (and re-written) prior to incurring a block erasure.

An $(n,k,t)_q$ flash code $\C$ is a coding scheme for storing $k$ information bits in $n$ cells in such a way that any sequence of up to $t$ writes
can be accommodated without a block erasure. The total number of available level transitions in $n$ cells is $n(q{-}1)$, and the \emph{write deficiency} of $\C$, defined\linebreak[4.0] as $\delta(\C) = n(q{-}1)-t$, is a measure of how close the code comes to perfectly utilizing all these transitions. In this paper, we show a construction of flash codes with write deficiency $O(qk\log k)$ if $q \ge \log_2k$, and at most \smash{$O(k\log^2 k)$} otherwise.

An $(n,r,\ell,t)_q$ buffer code is a coding scheme for storing a buffer of $r$ $\ell$-ary symbols such that for any sequence of $t$ symbols it is possible to successfully decode the last $r$ symbols that were written. We improve upon a previous upper bound on the maximum number of writes $t$ in the case where there is a single cell to store the buffer. Then, we show how to improve a construction by Jiang et al. that uses multiple cells, where $n\geq 2r$.
\end{abstract}

\section{Introduction}\label{sec:introduction}
Flash memories are, by far, the most important type of nonvol\-atile computer memory in use today. Flash devices are employ\-ed~widely
in mobile, embedded, and mass-storage applications, and the growth in this sector continues at a staggering pace.

A flash memory consists of an array of floating-gate \emph{\dfn cells},
organized into \emph{\dfn blocks} (a typical block contains about $2^{20}$ cells).
The level or ``state'' of a cell is a function
of the amount of charge (electrons) trapped within it. In
\emph{multilevel flash cells}, voltage is quantized to $q$ discrete
threshold values; consequently the level of each cell can be modeled
as an integer~in~the range $0,1,\ldots,q{-}1$. Nowadays, the parameter $q$ itself can
range from $q=2$ (the conventional two-state case) up to $q = 16$ and it can reach even higher values~\cite{GT05}.
The most conspicuous property of flash-storage technology is its
inherent asymmetry between cell programming (charge placement) and
cell erasing (charge removal). While adding charge to a single cell
is a fast and simple operation, removing charge from a cell is very
difficult. In fact, flash technology {does not allow} a~single
cell to be erased --- rather, {only entire blocks} can
be erased. Such \emph{\dfn block erasures} are not only
time-consuming, but also degrade the physical
quality of the memory. For example, a typical block in a multilevel
flash memory can tolerate only about $10^4$ or even fewer erasures before it
becomes unusable, and as such the lifetime and performance of the memory is highly
correlated with the frequency of block erasure operations.
Therefore, it is of importance to design coding schemes
that maximize the number of times information stored in
a flash memory can be written (and re-written) prior to
incurring a block erasure.

Such coding schemes --- known as \emph{\dfn floating codes} (or \emph{\dfn flash codes}) and \emph{\dfn buffer codes} --- were recently introduced in~\cite{J07,JBB07,BJB07}. Since then, several more papers on this subject have appeared in the literature~\cite{FLM08,JBB10,JB08,JLSB09,MSVWY09,YVSW08}. It should be pointed out that flash codes and buffer codes can be regarded as examples of memories with constrained source, which were described in~\cite{JLSB09}. Yet another example of such codes are the write-once memory (WOM) codes~\cite{FS84,RS82,CGM86}, that have been studied since the early 1980s. In fact, flash codes may be regarded as a generalization of WOM-codes. Slightly different and yet very related are the rank modulation codes~\cite{JMSB10,JSB10}. In rank modulation, the information is not stored according to the exact cell levels but rather by the cell permutation which is derived from the ordering of these levels.

An $(n,k,t)_q$ flash code $\C$ is a coding scheme for
storing $k$ information bits in $n$ flash-memory cells,
with $q$ levels each, in such a~way that any sequence
of up to $t$ writes can~be~accommodated without incurring
a block erasure. In the literature on flash codes,
a \emph{\dfn write} is always a bit-write --- that is,
a change $0 \to 1$ or $1 \to 0$ in the value of
one of the $k$ information bits. Observe that in order to accommodate such a write,
at least~one of the $n$ cells must transition from
a lower level to a~higher level (since a cell's level,
determined by its charge, can only increase).
On the other hand, the total number of available~level
transitions in $n$ flash cells is $n(q{-}1)$. Thus, throughout this
paper, we characterize the performance of a flash
code $\C$ in terms of its \emph{\dfn write deficiency}, defined as
$\delta(\C) \,{=}\; n(q{-}1)\,{-}\,t$. 
According to the foregoing discussion, $\delta(\C)$ is~a~measure~of
how close $\C$ comes to perfectly utilizing all the available
cell-level transitions: exactly one per write. 
The primary goal in designing flash codes can thus be expressed as
\emph{minimizing deficiency}.

What is the smallest possible write deficiency
$\delta_q(n,k)$~for~an $(n,k,t)_q$ flash
code, and how does it behave asymptotically~as the code parameters
$k$ and $n$ get large? The best-known {lower bound}, due to Jiang, Bo\-hos\-sian, and Bruck \cite{JBB07}, asserts that
\be{lower-bound}
\delta_q(n,k)
\, \ge \
\frac{1}{2}\bigl(q-1\bigr) \min\{n,k{-}1\}
\ee
How closely can this bound be approached by code constructions? It appears that the answer to this question depends on the relationship between $k$ and $n$. In this paper, we are concerned mainly with the case where both $k$ and $n$ are large, and $n$ is much larger than $k$ (in particular, $n \ge k^2$). In Section\,\ref{sec5}, we consider the case where $k/n$ is a constant. At the other end of the spectrum, the case $k > n$ has been studied in~\cite{JLSB09}.

The first construction of flash codes for large $k$ was~reported by Jiang and Bruck~\cite{JB08,JBB10}. In this construction, the $k$ information bits are partitioned into $m_1 \,{=}\, k/k'$ subsets of $k'$ bits~each~(with $k' \le 6$) while the memory cells are subdivided into $m_2 \ge m_1$ groups of $n'$ cells each. Additional memory cells (called \emph{\dfn index cells}) are set aside to indicate for each subset of $k'$ bits which group of $n'$ memory cells is used to store them. The deficiency of the resulting flash codes is \smash{$O(\sqrt{qn})$}. Note that for $n \ge k$, the lower bound on write deficiency in \eq{lower-bound} behaves as $\Omega(qk)$, and thus does not depend on $n$. Consequently, the gap between the Jiang-Bruck construction~\cite{JB08} and the lower bound could~be arbitrarily large, especially when $n$ is much larger than $k$.

In~\cite{YVSW08}, a different construction of flash codes was proposed. These codes are based upon representing the $n$ memory cells as a high-dimensional array, and achieve a write deficiency of $O(qk^2)$. Crucially, the deficiency of these codes does \emph{not} depend on $n$. Nevertheless, there is still a significant gap between $O(qk^2)$ --- which is the best currently known deficiency result --- and the lower bound of $\Omega(qk)$.

In this paper, we present a new construction of flash codes
which reduces the gap between the upper and lower bounds
on write deficiency to a factor that is \emph{logarithmic}
in the number of information bits $k$.
This result is arrived at in several stages.
As a starting point, we use the ``indexed'' flash codes
of Jiang and Bruck~\cite{JB08}. In Section\,\ref{sec3},
we develop new encoding and decoding procedures for such
codes that eliminate the need for index cells in the
Jiang-Bruck construction~\cite{JB08}. The write
deficiency achieved thereby 
is $O(qk^2)$, which coincides with the main result
of~\cite{YVSW08}. When the encoding procedure developed
in Section\,\ref{sec3} reaches its limit, there are
still potentially numerous unused cell-level transitions.
In Section\,\ref{sec4}, we show how to take advantage of these
transitions in order to accommodate even more
writes. To this end, we introduce a new indexing
scheme, which is invoked only after the encoding method of
Section\,\ref{sec3} reaches its limit. Thereupon, we extend this idea
recursively, through
$\ceil{\log_2\!k}$ different indexing stages.
This leads to a result, established in \Tref{main}, stating that
\be{bounds}
\Omega\bigl(qk\bigr)
\: \leq \:
\delta_q(n,k)
\: \leq \:
O\bigl(\max\{q,\log_2\! k\}\, k\log k \bigr)
\ee
for all $n \ge k^2$, where the upper bound is achieved constructively by the flash codes described in Section\,\ref{sec4}.
In Section~\ref{sec5}, we present and discuss constructions of flash codes for the case where the number of memory cells $n$ is not significantly larger than the number of bits $k$.

The other type of codes we discuss in this paper are the buffer codes. An $(n,r,\ell,t)_q$ buffer code is a coding scheme for storing a buffer of $r$ $\ell$-ary symbols such that for any sequence of $t$ symbol writes, it is possible to successfully decode the last $r$ symbols that were written without a block erasure. Given a buffer of $r$ $\ell$-ary symbols that has to be stored in $n$ $q$-ary cells, the goal is to maximize the number of writes $t$.

In Section~\ref{sec:buffer codes}, we formally define buffer codes. Then, we study two extreme cases where the number of cells is either one or very large. For the former case, Jiang et al. gave in~\cite{BJB07,JBB10} a construction as well as an upper bound on the number of writes. Their construction works for $n=1, \ell=2$ and guarantees $t=\left\lfloor\frac{q}{2^{r-1}}\right\rfloor + r-2$ writes. The upper bound stated in~\cite{BJB07,JBB10} for $n=1$ asserts that $$t\leq \left\lfloor\frac{q-1}{\ell^r-1}\right\rfloor\cdot r+\left\lfloor\left((q-1)\bmod(\ell^r-1)+1\right)\right\rfloor.$$  We will show how to improve this bound such that for $q\geq \ell^r$, $$t\leq\left\lfloor\frac{q-\ell^r}{\frac{1}{r}\sum_{d|r}\varphi(\frac{r}{d})\ell^d}\right\rfloor+r,$$
where $\varphi$ is Euler's $\varphi$ function.

If the buffer is binary ($\ell=2$) and the number of cells is significantly larger than the buffer size $r$, then a trivial upper bound on the number of writes $t$ is $n(q-1)$. 
Jiang et al. showed in~\cite{BJB07,JBB10} how to achieve $t=(q-1)(n-2r+1)+r-1$ writes. Assume that $q=2$, then the number of writes is $n-r$ and after the $i$-th write, the buffer is stored between cells $i+1$ and $i+r$. If $q>2$, then the cell levels are used layer by layer, where first only levels zero and one are used, then one and two, and so on. In the transition from one layer to another, first the buffer is copied and stored in the new layer and then more writes are allowed. Thus, this construction allows $n-r$ writes on the first layer and $n-2r+1$ more writes in all other layers, so the total number of writes is $t=n-r+(q-2)(n-2r+1)=(q-1)(n-2r+1)+r-1$. We will show how to improve this construction such that in every transition between layers, the buffer is stored cyclically in the cells and thus is not copied as before. This improves the number of writes to $(q-1)(n-r)$.

\section{Preliminaries and Flash Codes Definition}\label{sec:preliminaries}
Let us now give a precise definition of {flash codes}
that were introduced in the previous section.
We~use~$\{0,1\}^k$ to denote the set of
binary vectors of length $k$, and refer to the elements of this
set as \emph{\dfn information vectors}. The set of possible levels for
each cell is denoted by $\cA_q = \{0,1,\ldots,q{-}1\}$ and thought
of as a subset of the integers. The $q^n\!$ vectors of length $n$
over $\cA_q$ are called \emph{\dfn cell-state vectors}. With this notation,
any flash code $\C$ can be specified in~terms~of two
functions: an encoding map $\cE$ and a decoding map $\cD$.
The \emph{\dfn decoding map} $\cD\!: \cA_q^n \to \Ftk$ indicates
for each cell-state vector $\bfx \,{\in}\, \cA_q^n$ the
corresponding information vector. In~turn, the \emph{\dfn encoding map}
\smash{$\cE\!: \{0,1,\ldots,k{-}1\} {\times} \cA_q^n \to \cA_q^n \cup \{\E\}$}
assigns to every index $i$ 
and cell-state vector \smash{$\bfx \,{\in}\, \cA_q^n$},\,
another cell-state vector $\bfy = \cE(i,\bfx)$ such that
$y_j \,{\ge}\, x_j$ for all $j$ and $\cD(\bfy)$ differs
from $\cD(\bfx)$ \emph{only} in the $i$-th position.
If no such $\bfy \,{\in}\, \cA_q^n$ exists,
then $\cE(i,\bfx) = \E$ indicating that block
erasure is required. To~boot\-strap the encoding
process, we assume that the initial state of the
$n$ memory cells is $(0,0,\dots,0)$.
Henceforth, iteratively applying the encoding map, 
we can determine how \emph{any sequence} of transitions
$0 \to 1$ or $1 \to 0$ in the $k$ information bits maps
into a sequence of cell-state vectors, eventually terminated
by the block erasure. This leads to the following
definition.
\begin{definition}
An $(n,k)_q$ flash code\/ $\C(\cD,\cE)$
\emph{\dfn guarantees $t$ writes} if for all
sequences of up to $t$ transitions\, $0 \to 1$
or\, $1 \to 0$~in~the $k$ information bits,
the encoding map $\cE$ does not produce the block
erasure symbol\/ $\E$. If so, we say that\/ $\C$
is an $(n,k,t)_q$ code, and define the
\emph{\dfn deficiency of $\C$} as
$\delta(\C) \,{=}\; n(q{-}1)\,{-}\,t$.
\end{definition}

In addition to this 
definition, we will also
use the following terminology. Given a vector
$\bfx = (x_1,x_2,\dots,x_m)$~over $\cA_q$,~we define
its {\dfn weight} as
$\wt(\bfx) = x_1 + x_1 + \cdots + x_m$
(where the addition is over the integers), and
its {\dfn parity} as $\wt(\bfx)\!\!\mod 2$.

\section{Two-Bit Flash Codes}\label{sec:two bits}
In this section, we present a construction of flash codes that uses $n$ $q$-ary cells to store $k=2$ bits. In~\cite{JBB07}, a construction with these parameters was presented and was shown to be optimal. The construction we present in this section will be proved to be optimal as well and we believe that it is more intuitive.

In this construction, the leftmost and rightmost cells correspond to the first and second bit, respectively. When rewriting, assume the first bit changes its value, then the leftmost cell of level less than $q-1$ is increased by one level. Similarly, whenever the second bit changes its value, the rightmost cell of level less than $q-1$ is increased by one level. 
In general, the cell-state vector has the following form: $$(q-1,\ldots,q-1,x_i,0,\ldots,0,x_j,q-1,\ldots,q-1),$$ where $0< x_i,x_j\leq q-1$. This principle repeats itself until only one cell is left with level less than $q-1$. Then, this cell is used to store two bits according to its residue modulo $4$. If this residue is $0,1,2,3$ then the value of the bits is $(v_1,v_2) = (0,0),(1,0),(0,1),(1,1)$, respectively. The construction is presented for odd values of $q$ and we will discuss later how to modify it for even values as well. In the remainder of the paper, these maps are described algorithmically, using (C-like) pseudo-code notation.\\
\noindent
\textbf{Decoding map $\cD_{2B}$\kern1pt:}
The input to this map is a cell-state vector $\bfx = (x_1,x_2,\ldots,x_n)$. The output is the corresponding two-bit information vector $(v_1,v_2)$.
\codemode{%
~$i_1$\:=\:find\_left\_cell($y_1,y_2,\ldots,y_n$);\\[0ex]
~$i_2$\:=\:find\_right\_cell($y_1,y_2,\ldots,y_n$);\\[0ex]
~if($i_2$\;==\;$0$)\,\; /\!/\;\textrm{\sl all cells are full}\\[0ex]
~\Copen\;$v_{1}$\;=\;$q$\;-\;$1$(mod $2$);\;$v_{2}$\;=\;$\lfloor$(($q$\;-\;$1$)(mod $4$))/2$\rfloor$;\,\Cclose\\[0ex]
~if\;($i_1$\;==\;$i_2$)    /\!/\;\textrm{\sl there is only one non-full cell}\\[0ex]
~\Copen\;$v_{1}$\;=\;$y_{i_1}$(mod $2$);\;$v_{2}$\;=\;$\lfloor$($y_{i_1}$(mod $4$))/2$\rfloor$;\,\Cclose\\[0ex]
~if\;($i_1$\;!=\;$i_2$)    /\!/\;\textrm{\sl there are at least two non-full cells}\\[0ex]
~\Copen\;$v_{1}$\;=\;$y_{i_1}$(mod $2$);\;$v_{2}$\;=\;$y_{i_2}$(mod $2$);\,\Cclose\\[0ex]
}
\noindent
\textbf{Encoding map $\cE_{2B}$\kern1pt:}
The input to this map is a cell-state vector $\bfx = (x_1,x_2,\ldots,x_n)$, and an index $j\in\{1,2\}$ of the bit that has changed. Its output is either a new cell-state vector $\bfy = (y_1,y_2,\ldots,y_n)$ or the erasure symbol $\E$.\hspace*{3ex}
\codemode{%
~($y_1,y_2,\ldots,y_n$)\;=\;($x_1,x_2,\ldots,x_n$); \\[0ex]
~$i_1$\:=\:find\_left\_cell($y_1,y_2,\ldots,y_n$);\\[0ex]
~$i_2$\:=\:find\_right\_cell($y_1,y_2,\ldots,y_n$);\\[0ex]
~if($i_2$\;==\;$0$) return\;$\E$;\\[0ex]
~if\;($i_1$\;==\;$i_2$)    /\!/\;\textrm{\sl there is only one non-full cell}\\[0ex]
~\Copen\;if($j$\;==\;$2$)\;$a$\;=\;$2$;\;\\[0ex]
~\hspace{0.4ex} else\;\,$a$\;=\;$j$\;+\;$2\cdot$($y_{i_1}$(mod $2$));\\[0ex]
~\hspace{0.4ex} if($y_{i_1}$\;+\;$a$\;>\;$q$\;-\;$1$) return\;$\E$;\\[0ex]
~\hspace{0.4ex} else\;\Copen\;$y_{i_1}$\;=\;$y_{i_1}$\;+\;$a$;\,return;\,\Cclose\;\Cclose\\[0ex]
~\hspace{0.0ex}$y_{i_j}$\:=\:$y_{i_j}$\;+\;$1$;\;\\[0ex]
~if\;(($i_2$\;-\;$i_1$\;==\;$1$)\;$\land$\;($y_{i_j}$\;==\;$q$\;-\;$1$))\\[0ex]
~\Copen\;$v_{i_j}$\;=\;$0$;\;$v_{i_{3-j}}$\;=\;$y_{i_{3-j}}$(mod $2$);\\[0ex]
~\hspace{1.6ex}$a$\:=\:$2\cdot v_2$\;+\;$v_1$\;-\;($y_{i_{3-j}}$(mod $4$));\\[0ex]
~\hspace{1.6ex}if($a$\;<\;$0$) $y_{i_{3-j}}$\;=\;$y_{i_{3-j}}$\;+\;$4$\;+\;$x$;\\[0ex]
~\hspace{1.6ex}else $y_{i_{3-j}}$\;=\;$y_{i_{3-j}}$\;+\;$a$;\;\Cclose\\[0ex]
}

The function {\code find\_left\_cell($y_1,y_2,\ldots,y_n$)} finds the leftmost cell of level less than $q-1$ and if there is not such a cell then it returns $n+1$. Similarly, the function {\code find\_right\_cell($y_1,y_2,\ldots,y_n$)} finds the rightmost cell of level less than $q-1$ and if there is not such a cell then it returns $0$. The notation $y_{i_j}$ stands for the variable $y_{i_1}$ in case $j=1$, and $y_{i_2}$ if $j=2$. The same rule applies to $y_{i_{3-j}}$. The symbol $\land$ stands for the logical operator ``and''. The next theorem proves the number of writes this construction guarantees.
\begin{theorem}
If there are $n$ $q$-level cells and $q$ is odd, then the code \/ $\C(\cD_{2B},\cE_{2B})$ guarantees at least $t=(n-1)(q-1)+\left\lfloor\frac{q-1}{2}\right\rfloor$ writes before erasing.
\end{theorem}
\begin{proof}
As long as there is more than one cell of level less than $q-1$, the weight of the cell-state vector increases by one on each write. This may change only after at least $(n-1)(q-1)$ writes. Assume that there is only one cell of level less than $q-1$ after $s=(n-1)(q-1)+k$ writes, where $k\geq 0$, and call it the $i$-th cell. Starting this write, the different residues modulo $4$ of the $i$-th cell correspond to the four possible two-bit information vector $(v_1,v_2)$. Therefore, on the $s$-th write, we also need to increase the level of the $i$-th cell so it will correspond to the correct information vector on this write. For all succeeding writes, if the second bit changes then the $i$-th cell increases by two levels. If the first bit changes from $0$ to $1$ then the $i$-th cell increases by one level and otherwise by three levels. Therefore, if there are $m$ more writes and $v_1=0$ then the $i$-th cell increases by at most $2m$ levels, and if there are $m$ more writes and $v_1=1$ then the $i$-th cell increases by at most $2m+1$ levels.

Let us consider all possible values of $k$ and the information vector $(v_1,v_2)$ on the $s$-th write in order to calculate the number of guaranteed writes before erasing. Note that on the $s$-th write $(v_1+v_2)\equiv s(\bmod~2)$. Furthermore, since $q$ is odd, the value of the bit that is written changes from one to zero because it reaches level $q-1$, and thus the other bit has value $k(\bmod~2)$.
\begin{enumerate}
    \item Assume $k(\bmod~4) = 0$, then $(v_1,v_2)=(0,0)$ and the level of the $i$-th cell does not increase on the $s$-th write. Since $v_1=0$, after $m$ writes the cell increases by at most $2m$ levels. Hence, there are at least $\frac{q-1-k}{2}$ more writes and the total number of writes is at least \vspace{-2ex}

        \begin{small}$$(n-1)(q-1) + k +\frac{q-1-k}{2} \geq (n-1)(q-1) + \frac{q-1}{2}.$$\end{small}
    \item Assume $k(\bmod~4) = 1$, then $(v_1,v_2)=(1,0)$ or $(v_1,v_2)=(0,1)$. If $(v_1,v_2) = (1,0)$ then on the $s$-th write the $i$-th cell does not increase its level and after $m$ writes its level increases by at most $2m+1$ levels. If $(v_1,v_2) = (0,1)$ then the $i$-th cell increases by one level and after $m$ writes its level increases by at most $2m$ more levels. Hence, in both cases there are at least $\frac{q-2-k}{2}$ more writes. Together we get that the total number of writes is at least \vspace{-4.5ex}

        \begin{small}$$(n-1)(q-1) + k +\frac{q-2-k}{2} \geq (n-1)(q-1) + \frac{q-1}{2}.$$\end{small}
    \item Assume $k(\bmod~4) = 2$, then $(v_1,v_2)=(0,0)$ and the $i$-th cell increases by two levels on $s$-th write. Since $v_1=0$, after $m$ more writes the cell increases by at most $2m$ levels and hence there are at least $\left\lfloor(q-1-(k+2))/2\right\rfloor$ more writes, where $k\geq 2$. Therefore, the total number of write is at least \vspace{-2ex}

        \begin{small}$$(n-1)(q-1) + k + \frac{q-3-k}{2}  \geq (n-1)(q-1) + \frac{q-1}{2}.$$\end{small}
    \item Assume $k(\bmod~4) = 3$, then $(v_1,v_2)=(1,0)$ or $(v_1,v_2)=(0,1)$. If $(v_1,v_2) = (1,0)$ then on the $s$-th write the $i$-th cell increases by two levels and after $m$ more writes it increases by at most $2m+1$ levels. If $(v_1,v_2) = (0,1)$ then the $i$-th cell increases by three levels and after $m$ more writes it increases by at most $2m$ more levels. Hence there are at least $\frac{q-4-k}{2}$ more writes, where $k\geq 3$. Thus, the total number of writes is at least \vspace{-2ex}

        \begin{small}$$(n-1)(q-1) + k + \frac{q-4-k}{2}  \geq (n-1)(q-1) + \frac{q-1}{2}.$$\end{small}
\end{enumerate}
In any case, the guaranteed number of writes is
$(n-1)(q-1)+\left\lfloor\frac{q-1}{2}\right\rfloor$.
\end{proof}

For even values of $q$, the construction is very similar. As long as there is more than one cell of level less $q-1$ we follow the same rules for the encoding. For the decoding, since $q-1$ is no longer even, the value of $v_1$ is the parity of the cells $1,\ldots,i_1$, where $i_1$ is the leftmost cell of value less $q-1$. The value of $v_2$ is the parity of the cells $i_2,i_2+1,\ldots,n$, where $i_2$ is the rightmost cell of value less $q-1$. If there is only one cell left, then it represents a value of two bits as before according to its residue modulo $4$. If the the index of the last available cell is $i$ then
\begin{align*}
& v_1 = (i-1 + y_i) (\bmod 2), & \\
& v_2 =((n-i) + \left\lfloor (y_i(\bmod 4))/2 \right\rfloor)(\bmod 2). &
\end{align*}
Also, the last cell does not reach level $q-1$ so it is always possible to distinguish what the last cell is. We omit the tedious details as the proof is similar to the case where $q$ is odd.

\section{Index-less Indexed Flash Codes}\label{sec3}
Our point of departure is the family of so-called \emph{indexed flash codes}, due to Jiang and Bruck~\cite{JB08},
that were briefly described in Section\,\ref{sec:introduction}.
In this section, we eliminate the need for index cells
--- and, thus, the overhead associated with these cells --- in the
Jiang-Bruck construction~\cite{JB08}.
This is achieved by ``encoding'' the indices
into the {order} in which the cell levels are increased.

As in~\cite{JB08}, 
we partition the $n$ memory
cells into $m$ groups of $n'$ cells each.
However, while in
\cite{JB08} the value of $n'$ is more
or less arbitrary, in our construction $n'\! = k$.
We henceforth refer to such
groups of $n'\! = k$ cells as {\dfn blocks}
(though they are not related 
to the \emph{physical blocks} of floating-gate
cells which comprise the flash memory).
We will furthermore use, throughout this paper, the following
terminology. We say that:
\begin{itemize}
\item[\footnotesize$\blacktriangleright\!\!$]
a block is \emph{\dfn full} if all its cells are at level $q{-}1$;

\item[\footnotesize$\blacktriangleright\!\!$]
a block is \emph{\dfn empty} if all its cells are at level zero;

\item[\footnotesize$\blacktriangleright\!\!$]
a block is \emph{\dfn active} if it is neither full nor empty;

\item[\footnotesize$\blacktriangleright\!\!$]
a block is \emph{\dfn live} if it is not full (either active or empty).
\end{itemize}
In our construction, each block 
represents
\emph{exactly one bit}. This implies that the total number
of blocks, given by $m = \trunc{n/k}$, must be at least $k$,
which in turn implies $n \ge k^2$. If $n$ is not
divisible by $k$, the remaining
cells
are simply left unused. Finally, we
also assume that either $k$ is even or $q$ is odd. If this
is not the case, we can invoke the same construction with
$k$ replaced by $k\,{+}\,1$ (and the last 
bit permanently set to zero).

\looseness=-1
The key idea is that each block is used to encode not only
the current value of the bit that it represents, but also
\emph{which} of the $k$ bits it represents. The value of
the bit is simply the parity of the block. The index of
the bit is encoded in the \emph{order} in which the levels
of the $k$ cells are increased. For example, if the block
stores the $i$-th bit, first the level of the $i$-th cell
in the block is increased from $0$ to $q{-}1$ in response
to the transitions $0 \,{\to}\, 1$ and $1 \,{\to}\, 0$ in
the bit value. Then, the same procedure is applied to the
$(i{+}1)$-st cell, the $(i{+}2)$-nd cell, and so on,
with the indices $i\,{+}\,1,i\,{+}\,2,\dots$ interpreted cyclically
(modulo~$k$). This process is illustrated in the following
example.
\begin{example}
Suppose that $k = 4$ and $q = 3$. If a block represents
the first bit, then its cell levels will transition from
$(0,0,0,0)$ to $(2,2,2,2)$ in the following order:
$$
(0000) \to (1000) \to (2000) \to (2100) \to (2200)
$$
$$
\to (2210) \to (2220) \to (2221) \to (2222)
$$
On the other hand, for a block that represents the second bit,
the corresponding cell-writing order is given by:\vspace{-0.75ex}
$$
(0000) \to (0100) \to (0200) \to (0210) \to (0220)
$$
$$
\to (0221) \to (0222) \to (1222) \to (2222)
$$
The cell-writing orders for blocks that represent
the third and fourth bits are given, respectively, by\vspace{-0.75ex}
$$
(0000) \to (0010) \to (0020) \to (0021) \to (0022)
$$
$$
\to (1022) \to (2022) \to (2122) \to (2222)
$$
and
$$
(0000) \to (0001) \to (0002) \to (1002) \to (2002)
$$
$$
\to (2102) \to (2202) \to (2212) \to (2222)
$$
Note that, unless a block is full, it is always possible
to determine which cell was written first and, consequently,
which of the $k = 4$ bits this block represents.
\end{example}

We now provide a precise specification of an $(n,k)_q$
flash code $\C$ based upon this idea, in terms of a decoding
map $\cD_0$ and an encoding map $\cE_0$. 

\noindent
\textbf{Decoding map $\cD_0$\kern1pt:}
The input to this map is a cell-state vector
$\bfx = (\bfx_1|\bfx_2|\cdots|\bfx_m)$, partitioned
into $m$ blocks.
The output is the corresponding
information vector $(v_0,v_1,\dots,v_{k-1})$.
\codemode{%
~$(v_0,v_1,\dots,v_{k-1})$\;=\;$(0,0,\dots,0)$;\\[0.5ex]
~for\;($j$\:=\:$1$;\;$j \le m$;\;$j$\:=\:$j\,{+}\,1$)\\[-0.25ex]
~if\;(active($\bfx_j$))\\[-0.25ex]
~\Copen\;\,$\!i$\;=\;\,read\_index($\bfx_j$);\!
        $v_i$\:=\;\,parity($\bfx_j$);\,\Cclose
}

\noindent
\textbf{Encoding map $\cE_{0}$\kern1pt:}
The input to this map is a cell-state vector
$\bfx = (\bfx_1|\bfx_2|\cdots|\bfx_m)$, partitioned
into $m$ blocks of $k$ cells, and an index $i$ of
the bit that has changed. Its output is either a cell-state
vector $\bfy = (\bfy_1|\bfy_2|\cdots|\bfy_m)$
or the erasure symbol $\E$.\hspace*{3ex}
\codemode{%
~$(\bfy_1|\bfy_2|\cdots|\bfy_m)$\;=\;$(\bfx_1|\bfx_2|\cdots|\bfx_m)$;\\[0.50ex]
~for\;($j$\:=\:$1$;\;$j \le m$;\;$j$\:=\:$j\,{+}\,1$)\\[-0.25ex]
~if\;(active($\bfx_j$)\;$\land$\;(read\_index($\bfx_j$){}\,==\,\;$i$))\\[-0.25ex]
~\Copen\;write($\bfy_j$); break;\,\Cclose\\[-0.25ex]
~if\;($j$\;==\;$m+1$)    /\!/\;\textrm{\sl active block not found}\\[-0.25ex]
~for\;($j$\:=\:$1$;\;$j \le m$;\;$j$\:=\:$j\,{+}\,1$)\\[-0.25ex]
~if\;(empty($\bfx_j$))\;\Copen\:write\_new($i$,$\bfy_j$);\;break;\Cclose~\\[-0.25ex]
~if\;($j$\;==\;$m+1$)    /\!/\;\textrm{\sl no empty blocks remain}\\[-0.25ex]
~return\;$\E$;
}

To complete the specification of the flash code $\C(\cD_0,\cE_0)$,
let us elaborate upon all the functions used in the pseudo-code above.
The function {\code active($\bfx$)}, respectively {\code empty($\bfx$)},
simply determines whether the given block 
is active, respectiv\-ely empty. The function {\code parity($\bfx$)}\,computes
the parity of $\bfx$,
defined in Section\,\ref{sec:preliminaries}. Note that the parity
of a full block is always zero (since $k(q{-}1)$ is even, by assumption).
The function {\code read\_index($\bfx$)}\,computes the bit-index
encoded in an active block $\bfx = (x_0,x_1,\dots,x_{k-1})$. This
can be done as follows. Find all the zero cells in $\bfx$. Note
that these cells always 
form one cyclically contiguous run, say
$x_j,x_{j+1},\dots,x_{j+r}$ (where the indices are modulo $k$).
Then the index of the corresponding bit is $i = j+r+1~(\bmod k)$.
If there are no zeros in $\bfx$, there must be exactly one
cell, say $x_j$, whose level is strictly less than $q{-}1$. In
this case, the 
bit-index is $i = j+1~(\bmod k)$.
The function {\code write($\bfy$)}\,proceeds along similar
lines. Find the single cyclically contiguous run of zeros in
$(y_0,y_1,\dots,y_{k-1})$, say $y_j,y_{j+1},\dots,y_{j+r}$.
If $y_{j-1} \,{<}\, q{-}1$, increase $y_{j-1}$ by one; otherwise set
$y_j = 1$. If there are no zeros in $\bfy$, find the unique
cell $y_j$ such that $y_j \,{<}\, q{-}1$ and increase its level
by one. Finally, the function {\code write\_new($i$,$\bfy$)}\,simply
sets $y_i = 1$.
\begin{theorem}
\label{aux}
The write deficiency of the flash code\/ $\C(\cD_0,\cE_0)$
described above
is at most\, 
\be{th-aux}
(k\,{-}\,1)\Bigl((k\,{+}\,1)(q{-}1) \,-\, 1\Bigr) \ = \ O\bigl(qk^2\bigr)
\ee
\end{theorem}

\Proof 
Note that at each instance, at most $k$ of the $m$ blocks are active.
The encoding map $\cE_0(i,\bfx)$ produces the
symbol $\E$ when there are no more empty blocks, and none of the
active blocks represents the $i$-th bit. In the worst case,
this may occur when there are $k-1$ active blocks, each using just
one cell level. This contributes $(k\,{-}\,1)\bigl(k(q{-}1) - 1\bigr)$
unused cell levels.
In addition, there are at most $k-1$ cells that are unused due
to the partition into $m = \trunc{n/k}$ blocks of exactly $k$ cells.
These contribute at most $(k-1)(q{-}1)$ unused cell levels.
\qed\pagebreak[3.99]

\section{Nearly Optimal Construction}\label{sec4}
It is apparent from the proof of \Tref{aux} that the deficiency
of the flash code $\C(\cD_0,\cE_0)$, constructed in Section\,\ref{sec3},
is due primarily to the following: when writing stops, there may
remain potentially large amount of unused cell levels.
The key idea developed in this section is to
\emph{continue writing} after the encoding map $\cE_0$
produces the erasure symbol $\E$, 
utilizing those cell levels that are left unused by $\cE_0$.
Obviously, it is \emph{not} possible to continue writing
using the same encoding and decoding maps.
However, it may
be possible to do so if, at the point when $\cE_0$
produces the erasure symbol $\E$, we {switch} to
a \emph{different encoding procedure}, say $\cE_1$.
In fact, this idea can be applied iteratively:
once $\cE_1$ reaches its limit, we will transition
to another encoding map $\cE_2$, then yet another map $\cE_3$,
and so on.

Assuming that $k \equiv 0\! \pmod 4$,
here is one way to continue writing after the encoding map
$\cE_0$ has been exhausted.
When $\cE_0$ produces the erasure symbol $\E$,
we say that the \emph{first stage} of encoding is over 
and transition to the \emph{second stage}, as follows.
First, we re-examine the cell-state vector
$\bfx = (\bfx_1|\bfx_2|\cdots|\bfx_m)$
and re-partition it into $2m = 2\trunc{n/k}$ blocks
of {$k/2$ cells} each. Most of these smaller blocks
will already be full, but we may find some $m_1$ of
them that are either empty or active (live). Observe that
$m_1 \,{\le}\,\, 2(k\,{-}\,1)$ since at the end of the first
stage, there are at most $k-1$ active blocks of $k$
cells, and each of them produces at most two live (non-full)
blocks of $k/2$ cells.

If $m_1\! \ge k$, we can
continue writing as follows. Once again, each of
the $m_1$ blocks will represent exactly one bit;
as before, the value of this bit is determined
by the parity of the block. As part of the transition
from the first stage to the second stage, we record
the current information vector $(v_0,v_1,\dots,v_{k-1})$
in the first $k$ of the $m_1$ live blocks, say
$\bfx_1,\bfx_2,\dots,\bfx_{k}$. To this end,
whenever {\code parity($\bfx_i$)}\,$\ne v_{i-1}$,
we increase the level of one of the cells in $\bfx_i$ by one;
otherwise, we leave $\bfx_i$ as is.

Since the blocks now have $k/2$ cells rather than $k$ cells,
it is no longer possible to encode in each block
\emph{which} of the $k$ information bits it represents.
Therefore, we set aside for this purpose
\smash{$2(k{-}1)\ceil{\log_q(k{+}2)}$} \emph{index cells}
(that are not used during the first stage).
These cells are 
partitioned into $2(k{-}1)$ blocks of \smash{$\mu = \ceil{\log_q(k{+}2)}$}
cells each, which we call {\dfn index blocks}.
Henceforth, it will be convenient to refer to
the blocks of $k/2$ cells as {\dfn parity blocks},
in order to distinguish them from 
the index blocks. Initially, the first $k$ index blocks
\smash{$\bfu_1,\bfu_2,\dots,\bfu_{k}$}
are set so that \smash{$\bfu_i = {i}$} (in the base-$q$ number
system), which reflects the fact that the information bits
$v_0,v_1,\dots,v_{k-1}$ are stored (in that order)
in the first $k$ live parity blocks.
The next $m_1 \,{-}\, k$ index blocks are
set to $(0,0,\dots,0)$, thereby indicating that the
corresponding (live) parity blocks are available to store information
bits. The last $2(k{-}1) - m_1$ index blocks are set to
$(q{-}1,q{-}1,\dots,q{-}1)$ to indicate that the
corresponding parity blocks are full (in fact, nonexistent).
Finally, it is possible that in the process of
enforcing {\code parity($\bfx_i$)}\,$= v_{i-1}$ for the
first $k$ live parity blocks, some of these blocks become
full (this happens iff
$\wt(\bfx_i) = (k/2)(q{-}1) - 1$
and $v_i = 0$ at the end of the first stage,
since $k/2$ is even by assumption).
To account for this fact, we set the corresponding index
blocks to $(q{-}1,q{-}1,\dots,q{-}1)$. This completes
the transition from the first stage to the second stage,
which is invoked when the encoding map $\cE_0$
produces the erasure symbol $\E$.

Let us now summarize the foregoing discussion by giving
a concise algorithmic description of the transition
procedure.\\

\noindent
\textbf{Transition procedure $\cT_1$\kern1pt:}
Partition the memory into $2\trunc{n/k}$ parity
blocks of $k/2$ cells, and identify the $m_1 \le 2(k{-}1)$ parity
blocks $\bfx_1,\bfx_2,\ldots,\bfx_{m_1}$
that are not full. If $m_1 < k$, output the erasure symbol $\E$
and terminate. Otherwise, set the $2(k{-}1)$ {index} blocks
$\bfu_1,\bfu_2,\ldots,\bfu_{2k-2}$ as follows:
\be{set-index}
\bfu_i
\, = \:
\left\{\hspace{-0.5ex}
\begin{array}{lc@{\hspace{1ex}}l@{}}
  i & & \text{for $i = 1,2,\ldots,k$} \\[0.15ex]
  0 & & \text{for $i = k\,{+}\,1,k\,{+}\,2,\ldots,m_1$}\\[0.15ex]
q^\mu-1 & & \text{for $i = m_1{+}1,m_1{+}2,\ldots,2k\,{-}\,2$}
\end{array}
\right.\vspace*{0.75ex}
\ee
where \smash{$\mu \:{=}\, \ceil{\log_q(k{+}2)}$} is the number of
cells in each index block, \,then record the 
information vector $(v_0,v_1,\dots,v_{k-1})$ in
the first $k$ live parity blocks $\bfx_1,\bfx_2,\ldots,\bfx_{k}$,
as follows:\hspace*{2ex}
\codemode{%
~for\;($i$\:=\:$1$;\;$i \le k$;\;$i$\:=\:$i\,{+}\,1$)\\[-0.25ex]
~if\;(parity($\bfx_i$)$\ne v_{i-1}$)\\[-0.10ex]
~\Copen\;increment($\bfx_i$);
         if\;(full($\bfx_i$)\!)\:$\bfu_{i}$\:=\;$q^\mu-1$;\Cclose~
}
\noindent
The function {\code full($\bfx$)} determines
whether the given block $\bfx$ (which could be a parity
block or an index block) is full. The function
{\code increment($\bfx$)} increases by one the level of a cell
(does not matter which) in the given live block.\vspace{.250ex}

During second-stage encoding and decoding, we will need to figure
out for each active parity block $\bfx$ which of the $k$ information bits
it represents. To this end, we will have to find and read the index block
$\bfu$ that \emph{corresponds} to $\bfx$. How exactly is the correspondence
between parity blocks and index blocks established?
Note that, upon the completion of the transition procedure $\cT_1$,
there is the same number of {live} parity blocks and {live}
index blocks; moreover, the $j$-th \emph{live} index block corresponds
to the $j$-th \emph{live} parity block, for all $j$. The encoding
procedure will make sure that this correspondence is preserved
throughout the second stage: 
whenever a parity block becomes full, it
will make the corresponding index block
full as well.

\hspace*{-.5ex}We are now ready to present the encoding and
decoding \mbox{maps\hspace*{-.5ex}}
which are, again, specified in
C-like pseudo-code notation.\vspace{0.5ex}

\noindent
\textbf{Decoding map $\cD_1$\kern1pt:}
The input to this map is a cell-state vector
$\bfx = (\bfx_1|\bfx_2|\cdots|\bfx_{2m}|\!|\,\bfu_1|\bfu_2|\cdots|\bfu_{2k-2})$,
partitioned into $2m$ parity blocks, of $k/2$ cells each,
and $2(k{-}1)$ index blocks. The output is the 
information vector $(v_0,v_1,\dots,v_{k-1})$.
\codemode{%
~$(v_0,v_1,\dots,v_{k-1})$\;=\;$(0,0,\dots,0)$;\\[0.5ex]
~for\;($\ell$\:=\:$j$\:=\:$1$;\;$j \le 2m$;\;$j$\:=\:$j\,{+}\,1$)\\[-0.75ex]
~\Copen\\[-1.00ex]
~~~if\;(full($\bfx_j$)\!)\;continue;
         /\!/\;\textrm{\sl \!skip full blocks}\\[-0.55ex]
~~~while\;(full($\bfu_\ell$)\!)\;$\ell$\:=\:$\ell\,{+}\,1$;
         \hspace*{-1ex}/\!/\;\textrm{\sl \!skip full blocks~~~}\\[-0.25ex]
~~~$\,i$\;=\;$\bfu_\ell$; $\ell$\:=\:$\ell\,{+}\,1$;\\[-0.25ex]
~~~if\;($i \ne 0$)\;$v_{i-1}$\:=\;\,parity($\bfx_j$);\\[-1.5ex]
~\Cclose
}\vspace{.5ex}

\noindent
Given an index $i$ of the bit that has changed, the
encoding map $\cE_1$ first tries to find an active parity
block $\bfx$ that represents the $i$-th information bit.
If such a block is found, it is incremented and checked to see
if it is full (in which case the corresponding index block
is set to $q^\mu-1$). If not, another live parity block
is allocated to represent the $i$-th information bit.
If no more live parity blocks are available,

the erasure symbol $\E$
is returned.\vspace{0.5ex}

\noindent\looseness=-1
\textbf{Encoding map $\cE_1$\kern1pt:}
The input to this map is a cell-state vector
$\bfx = (\bfx_1|\bfx_2|\cdots|\bfx_{2m}|\!|\,\bfu_1|\bfu_2|\cdots|\bfu_{2k-2})$,
partitioned into $2m$ parity blocks 
and $2(k{-}1)$ index blocks, and an index $i$ of
the information bit that changed.\
Its output is either a cell-state vector 
$\bfy= (\bfy_1|\bfy_2|\cdots|\bfy_{2m}|\!|\,\bfu'_1|\bfu'_2|\cdots|\bfu'_{2k-2})$
or the symbol~$\E$.\hspace*{.25ex}
\codemode{%
~$(\bfy_1|\bfy_2|\cdots|\bfy_{2m})$\;=\;$(\bfx_1|\bfx_2|\cdots|\bfx_{2m})$;\\
~$(\bfu'_1|\bfu'_2|\cdots|\bfu'_{2k-2})$\;=\;$(\bfu_1|\bfu_2|\cdots|\bfu_{2k-2})$;
\\[0.50ex]
~for\;($\ell$\:=\:$j$\:=\:$1$;\;$j \le 2m$;\;$j$\:=\:$j\,{+}\,1$)\\[-0.75ex]
~\Copen\\[-1.00ex]
~~~if\:(full($\bfx_j$)\!)\;continue;\\[-0.65ex]
~~~while\;(full($\bfu_\ell$)\!)\;$\ell$\:=\:$\ell\,{+}\,1$;\\[-0.05ex]
~~~if\:($\bfu_\ell$\:==\;$i+1$)\\[-0.7ex]
~~~\Copen\\[-1.00ex]
~~~~~increment($\bfy_j$);\\[-.70ex]
~~~~~if\kern0.5pt(full($\bfy_j$)\!)\:$\bfu'_{\ell}$\:=\;$q^\mu-1$;\\[-.85ex]
~~~~~break;\\[-.75ex]
~~~\Cclose\\[-.75ex]
~~~else $\ell$\:=\:$\ell\,{+}\,1$;\\[-0.50ex]
~\Cclose\\[0.50ex]
~if\;($j$\;==\;$2m+1$)   /\!/\;\textrm{\sl active block not found}  \\[-0.50ex]
~for\;($\ell$\:=\:$j$\:=\:$1$;\;$j \le 2m$;\;$j$\:=\:$j\,{+}\,1$)\\[-0.75ex]
~\Copen\\[-1.00ex]
~~~if\:(full($\bfx_j$)\!)\;continue; \\[-0.65ex]
~~~while\;(full($\bfu_\ell$)\!)\;$\ell$\:=\:$\ell\,{+}\,1$;\\[-0.05ex]
~~~if\:($\bfu_\ell$\:==\;$0$)\\[-0.7ex]
~~~\Copen\\[-1.25ex]
~~~~~$\bfu'_\ell$\:=\;$i+1$;\\[-0.5ex]
~~~~~if\:(parity($\bfx_j$)$\ne v_{i}$)\,increment($\bfy_j$);\\[-0.78ex]
~~~~~if\:(full($\bfy_j$)\!)\:$\bfu'_{\ell}$\:=\;$q^\mu-1$;\\[-.82ex]
~~~~~break;\\[-.75ex]
~~~\Cclose\\[-.75ex]
~~~else $\ell$\:=\:$\ell\,{+}\,1$;\\[-0.50ex]
~\Cclose\\[0.50ex]
~if\;($j$\;==\;$2m+1$)
/\!/\;\textrm{\sl no more available live blocks~~}\\[-0.5ex]
~return\;$\E$;}\vspace{1.00ex}

Note that when the second encoding stage terminates, there are
at most $k-1$ parity blocks that are not full, comprising~at
most $k(k-1)/2$ cells (at most $k(k-1)(q{-}1)/2$ cell-levels).\vspace{.50ex}

Once the maps $\cD_1$ and $\cE_1$ are understood, it becomes
clear that the same approach can be applied iteratively.
The resulting flash code $\C^*$ will proceed, sequentially,
through $s$ different~encoding stages $\cE_0,\cE_1,\dots,\cE_{s-1}$,
where $s = \ceil{\log_2k}$. In describing this code, we shall
assume for the sake of simplicity that $k$ is a power of two,
that is $k = 2^s$. If not, the same code can be used to store
$2^s > k$ information bits,~of~which~the last $2^s - k$ are set
to zero. Note that this will not change the order of the resulting
write deficiency.

\hspace*{-2pt}To accommodate
the encoding maps $\cE_1,\cE_2,\dots,\cE_{s-1}$,~we~set
aside for \emph{each map} a batch of $2(k-1)$ index blocks,
with~each index block consisting of
\smash{$\mu = \ceil{\log_q(k{+}2)}$} cells.
The transition procedure $\cT_r$ which bridges between
the encoding maps $\cE_{r-1}$ and $\cE_r$
(for some $r \,{\in}\, \{2,3,\dots,s{-}1\}$)
is identical
to the transition procedure $\cT_1$, except for the
following differences:
\begin{description}
\item[~~~{\bf D1.}]
The $r$-th batch of index blocks is used; and

\item[~~~{\bf D2.}]
The parity blocks consist of $k/2^r$ cells each.
\end{description}
In addition to {\bf D1} and {\bf D2}, the decoding/encoding
maps~$\cD_r$~and $\cE_r$ differ from $\cD_1$ and $\cE_1$ in
that ``$2m$'' should be replaced by ``$2^rm$'' throughout,
where $m$ stands for $\trunc{n/k}$ as before. There are
no other differences.\vspace{-1.25ex}
\begin{theorem}
\label{th2}
For $s \,{=}\,\ceil{\log_2\!k}$, 
the write deficiency of the flash code\/ $\C^*\!$
defined by the sequence of decoding/encoding~maps
$\cD_0,\cD_1,\dots,\cD_{s-1}$
and\/ $\cE_o,\cE_1,\dots,\cE_{s-1}$
is\/
\smash{$O\bigl(qk\log^2 \!k/\!\log q\bigr)$}.\vspace*{-0.50ex}
\end{theorem}

\Proof
We consider the worst-case scenario for the number of cell
levels that are either unused or ``wasted'' in the overall
encoding procedure.
As before, there are at most $k-1$ cells that are unused due 
to the partition into $\trunc{n/k}$ blocks,~of~ex-actly $k$ cells each,
at the very first encoding stage.
These cells contribute at most $(q{-}1)(k-1)$ unused cell levels.
The index blocks for the $s-1$ encoding maps $\cE_1,\cE_2,\dots,\cE_{s-1}$
contain $2(k-1)(s-1)\mu$ cells altogether, thereby wasting at
most\vspace{-0.15ex}
\be{th2-eq1}
2(q-1)(k-1)(s-1) \ceil{\log_q(k{+}2)}
\: = \:
O\biggl(\frac{qk\log^2 \!k}{\log q}\biggr)
\ee
cell levels. In each of the $s-1$ transition procedures,
the situ\-at\-ion
{\code parity($\bfx_i$)${\ne}\,v_{i-1}$} can occur
at most $k$ times, and each time it occurs a single
cell level is wasted. Finally, as in \Tref{aux},
when the encoding process $\cE_o,\cE_1,\dots,\cE_{s-1}$
terminates there are at
most $k-1$ parity blocks that are not full and,
in the worst case, each of them uses just one cell
level. However, 
now these parity blocks contain only $\lceil k/2^{s-1}\rceil = 2$ cells
each, and thus contribute at most $(k-1)(2q-3)$
unused cell levels. Putting all of this together, we find
that at most\vspace{-0.20ex}
\be{th2-eq2}
(q{-}1)(k{-}1)\Bigl(2(s{-}1)\ceil{\log_q(k{+}2)}\,+\, 3\Bigr)
\ + \ k(s{-}1)
\vspace*{-.50ex}
\ee
cell levels are wasted or left unused. Clearly,
this~expression~is dominated by
\Eq{th2-eq1}, and thus bounded by
\smash{$O\bigl(qk\log^2 \!k/\!\log q\bigr)$}.~\qed\vspace{1.00ex}

For large $q$, the upper bound of \smash{$O\bigl(qk\log^2 \!k/\!\log q\bigr)$}
on the de\-ficiency of our scheme can be improved by using
a~more~efficient ``packaging'' of 
index blocks in the flash memory.
As before, we allocate a batch of $2(k-1)$ index blocks~to~each
encoding stage except $\cE_0$. But now, every index
block~will~occupy
\smash{$\mu' = \ceil{\log_2(k{+}2)}$}
cells rather than
\smash{$\mu = \ceil{\log_q(k{+}2)}$}
cells, and the indices will be written in binary rather than
in the base-$q$ number system. This allows index blocks that
correspond to successive encoding stages to be ``stacked on top
of each other'' in the same memory cells. Specifically, the
encoding stage $\cE_1$ will use only cell levels $0$ and $1$
to record the indices in its index blocks. Once this stage is
over, the index information recorded during $\cT_1$ and $\cE_1$
is no longer relevant, and the level of \emph{all} the
$2(k-1)\mu'$ cells in the $2(k-1)$ index blocks can be
raised to $1$. Thereafter, provided $q \ge 3$, the transition
procedure $\cT_2$ and the encoding map $\cE_2$ can use cell
levels $1$ and $2$ to record the relevant index information
in the \emph{same memory cells}. Proceeding in this manner,
we can accommodate up to $q-1$ batches of index blocks in
$2(k-1)\mu'$~memory cells. We shall refer to this indexing
scheme as \emph{stacked binary indexing} and denote the
resulting flash code by $\C'$.\vspace{-1.10ex}
\begin{theorem}
\label{main}
The write deficiency of the flash code\/ $\C'$
defined by the sequence of decoding/encoding~maps
$\cD_0,\cD_1,\dots,\cD_{s-1}$
and\/ $\cE_o,\cE_1,\dots,\cE_{s-1}$ that use stacked
binary indexing is at most
$O(qk\log k)$ if $q \ge \log_2k$, and at most
\smash{$O(k\log^2 k)$} otherwise.
\end{theorem}

\Proof
With stacked binary indexing, the number of cell~le\-vels
wasted in {all} the $2(k\,{-}\,1)(s\,{-}\,1)$ index blocks
is~at~most~~
\be{main-eq1}
2(q-1)(k-1)\left\lceil \frac{s-1}{q-1}\right\rceil \ceil{\log_2(k{+}2)}
\ee
Although for most values of $k$ and $q$ this is strictly
less than~\Eq{th2-eq1}, all the other terms in \Eq{th2-eq2} are still
dominated by \Eq{main-eq1}.
\qed\vspace{1.5ex}

\noindent
{\bf Remark.}
If we need to store $k$ \emph{symbols}, rather than bits,
over an alphabet of size $\ell > 2$, the same flash code
can still be used, with an appropriate interface. With the linear
WOM-code of~\cite{RS82}, the $\ell$-ary symbols can be represented
using $\ell\,{-}\,1$ bits in such a way that any symbol change
corresponds to a single 
bit transition.
The flash code $\C'$ can be now applied as is, and
the resulting write deficiency is
$O\bigl(\max\{q,\log_2\! k\ell\}\, k\ell \log k\ell \bigr)$.

\section{Flash Codes of Constant Rate}\label{sec5}
All of our results so far pertain to the case where $n \ge k^2$.
In this section, we briefly examine the situation where both~$k$~and
$n$ are large, while $k/n = R$ for some constant $R < 1$. Observe
that write deficiency $\delta(\C) = n(q{-}1)-t$ is \emph{not}
an appropriate figure of merit in this situation: a trivial code that 
guarantees $t=0$ writes achieves write deficiency $n(q{-}1) = k(q{-}1)/R$,
which is within a constant factor $2/R$ from the lower
bound~\Eq{lower-bound}.\hspace*{-3pt}\linebreak
Thus we will state our results in terms of the
guaranteed~number of writes $t$
rather than the write deficiency $\delta(\C)$.

If $q=2$, we can easily guarantee $\Omega(n/\kern-1pt\log k)$
writes~as~follows:
partition the $n$ cells into blocks of size
\smash{$\ceil{\log_2 \!k}$} and each time an information bit changes,
record its index in the next available block.
For $q > 2$, the same method guarantees about
$\lfloor{n/\kern-1pt\log_q \!k}\rfloor = \Omega(n\log q/\kern-1pt\log k)$
writes, but we can do better.~~

Let us 
partition the $n$ cells into two groups: the
\emph{index group} consisting of $n-k$ cells and the
\emph{parity group} consisting of $k$ cells. The index
group is then subdivided into \smash{$m\:\,{=}\trunc{(n{-}k)/s}$}
blocks, each consisting of \smash{$s = \ceil{\log_2 \!k}$} cells.
The writing proceeds in $q-1$ phases. During the first phase,
every time an information bit changes, its index is recorded
in binary (using cell levels $0$ and $1$) in the next available
index block. After $m$ writes, the first phase is over. We then
copy the $k$ information bits into the $k$ cells of the parity
group, and raise the level of all cells in the index group to $1$.
The second phase can now proceed using cell levels $1$ and $2$,
and recording changes in information bits relative to the values
stored in the parity group. At the end of the second phase, the
current values of the $k$ bits are recorded in the parity cells
using levels $1$ and $2$, and so on.
This simple coding scheme achieves
\be{last}
m(q-1)
\ = \
\frac{n(q{-}1)(1-R)}{\log_2 k}
\ = \
\Omega\biggl(\frac{nq}{\log k}\biggr)
\ee
writes (where the middle expression ignores ceilings/floors~by
assuming that $k$ is a power of two and that $n-k$ is divisible by
$\log_2k$). If $q$ is odd and $R \ge 0.415$,
we can do a little~better by using the ternary number system
(cell levels $0$,\,$1$,\,$2$) in both the index group and the
parity group. In this case, the size of the parity group is
\smash{$\lceil k/\kern-1pt\log_2\!3\rceil$} cells and $1-R$
in~\Eq{last}~can~be~re-placed by $(\log_2\!3 - R)/2$. Finally,
for all $R \ge 0.755$ and~\mbox{$q\,{-}\,1$} divisible by three,
the quaternary
alphabet is optimal, leading to a factor of $(2\,{-}\,R)/3$ rather
than $1\,{-}\,R$ in \Eq{last}.

\section{Buffer Codes}\label{sec:buffer codes}
Buffer codes were first presented by Bohossian et al. in~\cite{BJB07}. In this family of codes, a buffer of $r$ symbols has to be stored in $n$ flash memory $q$-ary cells. After each write, the last $r$ symbols that were written have to be recovered by the cell-state vector. The goal is to maximize $t$, the number of write symbols that the code guarantees without incurring a block erasure. In~\cite{BJB07,JBB10}, an upper bound and a construction are presented for the case where the buffer is stored in a single cell. It is also shown how to store a buffer where, $n$, the number of cells satisfies $n\geq 2r$.

\subsection{Buffer Codes Definition}
We refer to the set of vectors in $\{0,\ldots,\ell-1\}^r$ as \emph{\dfn buffer vectors}. Similarly to a flash code, a buffer code $\C$ is also specified by an encoding map $\cE$ and a decoding map $\cD$. The \emph{\dfn decoding map} $\cD\!: \cA_q^n \to \{0,\ldots,\ell-1\}^r$ assigns for each cell-state vector $\bfx \,{\in}\, \cA_q^n$ its buffer vector $\cD(\bfx)$. The \emph{\dfn encoding map} \smash{$\cE\!: \{0,\ldots,\ell-1\} {\times} \cA_q^n \to \cA_q^n \cup \{\E\}$} specifies for every symbol $a\in\{0,\ldots,\ell-1\}$ and cell-state vector \smash{$\bfx \,{\in}\, \cA_q^n$},\, another cell-state vector $\bfy = \cE(a,\bfx)$ such that $y_j \,{\ge}\, x_j$ for all $1\leq j\leq n$, $(\cD(\bfy))_1=a$ and for $2\leq i\leq r$, $(\cD(\bfy))_i=(\cD(\bfx))_{i-1}$. In case such a $\bfy \,{\in}\, \cA_q^n$ does not exist, then $\cE(i,\bfx) = \E$. 
\begin{definition}
An $(n,r,\ell,t)_q$ buffer code\/ $\C(\cD,\cE)$ \emph{\dfn guarantees $t$ writes} if for all sequences of up to $t$ symbol writes, the encoding map $\cE$ does not produce the block erasure symbol\/ $\E$.
\end{definition}

\subsection{Single-Cell Buffer Codes}
In this section, we discuss the case where there is a single cell ($n=1$) to store the buffer. A construction for this scenario where a binary buffer ($\ell=2$) is stored was given in~\cite{BJB07,JBB10}. This construction guarantees at least $t=\left\lfloor\frac{q}{2^{r-1}}\right\rfloor +r-2$ writes before a block erasure. An upper bound was given as well, which asserts that for every buffer code with one cell, the number of writes $t$ has to satisfy $$t\leq \left\lfloor\frac{q-1}{\ell^r-1}\right\rfloor\cdot r+\left\lfloor\left((q-1)\bmod(\ell^r-1)+1\right)\right\rfloor.$$ Let us show here another upper bound for such codes. 
\begin{theorem}\label{thm:buffer bound}
For any $(1,r,\ell,t)_q$ buffer code $\C$ such that $q\geq \ell^r$, $$t\leq\left\lfloor\frac{q-\ell^r}{\frac{1}{r}\sum_{d|r}\varphi(\frac{r}{d})\ell^d}\right\rfloor+r,$$ where $\varphi$ is Euler's $\varphi$ function.
\end{theorem}
\begin{proof}
Let $\C(\cD,\cE)$ be a $(1,r,\ell,t)_q$ buffer code.
After $i\geq 1$ writes, for each $\bfv\in\{0,1,\ldots,\ell-1\}^r$, let
\begin{align*}
S_i(\bfv)=\{x\ | \ & \textrm{there is a sequence of $j\leq i$ symbol } &  \\
& \textrm{writes ending in level $x$ and $\cD(x)=\bfv$}\}, &
\end{align*}
$m_i(\bfv) = \max_{x\in S_i(\bfv)}\{x\}$ is the maximum cell level that is possible to reach after $i$ symbol writes such that $\cD(m_i(\bfv))=\bfv$, and
$$M_i = \sum_{\bfv\in\{0,\ldots,\ell-1\}^r}|S_i(\bfv)|.$$ Clearly, for all $i\leq t$, $M_i\leq q-1$.  After $r$ writes, it is possible to reach any of the $\ell^r$ different buffer vectors and thus $M_r\geq \ell^r-1$.

Let $\cG_{\ell,r}$ be the $r$-th order $\ell$-ary de Bruijn graph~\cite{B46}. Its vertex set is $\cV_{\ell,r}=\{0,1,\ldots,\ell-1\}^r$ and its edge set is $\cE_{\ell,r}$. Let $\bfv_1,\bfv_2\in \{0,1,\ldots,\ell-1\}^r$ be two different buffer states. Note that if $(\bfv_1,\bfv_2)\in\cE_{\ell,r}$ and $m_i(\bfv_1)>m_i(\bfv_2)$ then $m_{i+1}(\bfv_2)>m_i(\bfv_2)$ and therefore, the value of $M_{i+1}$ increases by at least one level for every such an edge. In the de~Bruijn graph, every cycle has at least one edge $(\bfv_1,\bfv_2)\in\cE_{\ell,r}$ such that $m_i(\bfv_1)>m_i(\bfv_2)$. Therefore, the number of new unused levels is at least the number of disjoint vertex cycles in $\cG_{\ell,r}$. This number is known to be $\frac{1}{r}\sum_{d|r}\varphi(\frac{r}{d})\ell^d$~\cite{M72,G82}, and therefore $$t\leq\left\lfloor\frac{q-\ell^r}{\frac{1}{r}\sum_{d|r}\varphi(\frac{r}{d})\ell^d}\right\rfloor+r.$$
\end{proof}

\begin{lemma}
The bound in Theorem~\ref{thm:buffer bound} improves the bound in~\cite{BJB07} for $q\geq \ell^r$. That is,
\begin{align*}
& \left\lfloor\frac{q-\ell^r}{\frac{1}{r}\sum_{d|r}\varphi(\frac{r}{d})\ell^d}\right\rfloor +r & \\
& \leq \left\lfloor\frac{q-1}{\ell^r-1}\right\rfloor\cdot r + \left\lfloor\log_{\ell}\big(((q-1)\bmod (\ell^r-1))+1\big)\right\rfloor.&
\end{align*}
\end{lemma}
\begin{proof}
Note that $$\frac{1}{r}\sum_{d|r}\varphi\left(\frac{r}{d}\right)\ell^d \geq \frac{\ell^{r}+\ell\varphi\left(r\right)}{r},$$ and therefore
\begin{align*}
& \left\lfloor\frac{q-\ell^r}{\frac{1}{r}\sum_{d|r}\varphi(\frac{r}{d})\ell^d}\right\rfloor +r \leq \left\lfloor\frac{q-\ell^r}{\frac{\ell^{r}+\ell\varphi(r)}{r}}\right\rfloor +r & \\
& = \left\lfloor\frac{q-\ell^r}{\ell^{r}+\ell\varphi(r)}\cdot r\right\rfloor +r = \left\lfloor\frac{q+\ell\varphi(r)}{\ell^{r}+\ell\varphi(r)}\cdot r\right\rfloor. &
\end{align*}
If we denote $q-1=x(\ell^r-1)+y$, where $0\leq y\leq \ell^r-1$, then
\begin{align*}
& \left\lfloor\frac{q-\ell^r}{\frac{1}{r}\sum_{d|r}\varphi(\frac{r}{d})\ell^d}\right\rfloor +r \leq \left\lfloor\frac{q+\ell\varphi(r)}{\ell^{r}+\ell\varphi(r)}\cdot r\right\rfloor & \\
& = \left\lfloor\frac{x(\ell^r-1)+y+1+\ell\varphi(r)}{\ell^{r}+\ell\varphi(r)}\cdot r\right\rfloor & \\
& = \left\lfloor\frac{x(\ell^r+\ell\varphi(r))-x+y+1-(x-1)\ell\varphi(r)}{\ell^{r}+\ell\varphi(r)}\cdot r\right\rfloor & \\
& = xr + \left\lfloor\frac{-x+y+1-(x-1)\ell\varphi(r)}{\ell^{r}+\ell\varphi(r)}\cdot r\right\rfloor & \\
& \leq xr + \left\lfloor\frac{(y+1)r}{\ell^{r}}\right\rfloor. &
\end{align*}
Let us show that $\frac{(y+1)r}{\ell^r}\leq \log_{\ell}(y+1)$. That is, we show that $(y+1)\geq \ell^{\frac{(y+1)r}{\ell^r}}$ or $$\left((y+1)^{\frac{1}{y+1}}\right)^{\ell^r}\geq \ell^r.$$ The function $f(x)= x^{\frac{1}{x}}$ is monotonically decreasing for $x\geq~1$ and since $y\leq \ell^r-1$, we get $$\left((y+1)^{\frac{1}{y+1}}\right)^{\ell^r}\geq \left((\ell^r)^{\frac{1}{\ell^r}}\right)^{\ell^r} = \ell^r.$$
Putting these together we get
\begin{align*}
& \left\lfloor\frac{q-\ell^r}{\frac{1}{r}\sum_{d|r}\varphi(\frac{r}{d})\ell^d}\right\rfloor +r \leq xr + \left\lfloor\frac{(y+1)r}{\ell^{r}}\right\rfloor & \\
& \leq xr + \left\lfloor\log_{\ell}(y+1)\right\rfloor & \\
& = \left\lfloor\frac{q-1}{\ell^r-1}\right\rfloor\cdot r + \left\lfloor\log_{\ell}\big(((q-1)\bmod (\ell^r-1))+1\big)\right\rfloor. &
\end{align*}
\end{proof}

\subsection{Multiple-Cells Buffer Codes}
In~\cite{BJB07,JBB10}, a buffer code construction is given for $\ell=2$ and arbitrary $n,q,r$, where $n\geq 2r$. This construction guarantees $t=(q-1)(n-2r+1) +r-1$ writes. In this section, we show how to improve this construction such that the guaranteed number of writes is $t=(q-1)(n-r)$.

In the case where $q=2$, the construction in~\cite{BJB07,JBB10} guarantees $n-r$ writes. The encoding procedure is performed in such a way that after $i$ writes, $1\leq i\leq n-r$, the buffer is located between the $(i+1)$-st and $(i+r)$-th cells, where the first bit of the buffer memory is stored in the $(i+r)$-th cell and the last bit is stored in the $(i+1)$-st cell. If $q>2$, then the construction uses a ``layer by layer'' approach. That is, first the layer of levels $0$ and $1$ is used, then the layer of levels $1$ and $2$ is used, and so on. In the transition from the layer of levels $i-1$ and $i$ to the layer of levels $i$ and $i+1$, all the cells are first reset to level $i$ and the buffer is written in the new layer of levels $i$ and $i+1$. Then, it is possible to continue writing in this layer. Basically, on each layer, it is possible to write $n-r$ times. However, when a new layer is used, then first the buffer from the previous layer is copied and then it is written in the new layer. Hence, it is possible to have only $(n-2r+1)$ more writes in the new layer and thus the total number of writes is $$n-r + (q-2)(n-2r+1) = (q-1)(n-2r+1) +r-1.$$

The transition between these consecutive layers is not performed efficiently and our improvement here shows how it is possible to write $n-r$ times on each layer such that the total number of writes is $t=(q-1)(n-r)$. We first demonstrate how the construction works by the following example.
\begin{example}
In this example, we show how the last construction works for $n=11, q=3,\ell=2$ and $r=4$, so the number of writes is $2\cdot (11-4)=14$. The sequence
of bits to be written is $1,1,0,0,1,0,0,1,1,1,0,1,1,0$ and the writes are performed as follows. The underlined cells represent the cells that store the buffer on each write.
$$\begin{tabular}{|c|c|c|}
\hline  Written Bit & Buffer State & Cell State Vector
\\\hline \hline
      & $(0,0,0,0)$ & $(\underline{0,0,0,0},0,0,0,0,0,0,0)$ \\
  $1$ & $(0,0,0,1)$ & $(0,\underline{0,0,0,1},0,0,0,0,0,0)$ \\
  $1$ & $(0,0,1,1)$ & $(0,0,\underline{0,0,1,1},0,0,0,0,0)$ \\
  $0$ & $(0,1,1,0)$ & $(1,0,0,\underline{0,1,1,0},0,0,0,0)$ \\
  $0$ & $(1,1,0,0)$ & $(1,1,0,0,\underline{1,1,0,0},0,0,0)$ \\
  $1$ & $(1,0,0,1)$ & $(1,1,0,0,1,\underline{1,0,0,1},0,0)$ \\
  $0$ & $(0,0,1,0)$ & $(1,1,1,0,1,1,\underline{0,0,1,0},0)$ \\
  $0$ & $(0,1,0,0)$ & $(1,1,1,1,1,1,0,\underline{0,1,0,0})$ \\
  $1$ & $(1,0,0,1)$ & $(1,1,1,1,\underline{2},1,1,1,\underline{1,0,0})$ \\
  $1$ & $(0,0,1,1)$ & $(1,1,1,1,\underline{2,2},1,1,1,\underline{0,0})$ \\
  $1$ & $(0,1,1,1)$ & $(1,1,1,1,\underline{2,2,2},1,1,1,\underline{0})$ \\
  $0$ & $(1,1,1,0)$ & $(2,1,1,1,\underline{2,2,2,1},1,1,1)$ \\
  $1$ & $(1,1,0,1)$ & $(2,1,1,1,2,\underline{2,2,1,2},1,1)$ \\
  $1$ & $(1,0,1,1)$ & $(2,1,1,1,2,2,\underline{2,1,2,2},1)$ \\
  $0$ & $(0,1,1,0)$ & $(2,1,1,1,2,2,2,\underline{1,2,2,1})$ \\ \hline
\end{tabular}$$
\end{example}

Now we are ready to present the construction by specifying its encoding and decoding maps specification.

\noindent
\textbf{Decoding map $\cD_{\textmd{buf}}$\kern1pt:}
The input to this map is a cell-state vector $\bfx = (x_1,x_2,\ldots,x_n)$. The output is the corresponding information buffer vector $(v_1,v_2,\ldots,v_r)$.
\codemode{%
~$m$\:=\:max($x_1,x_2,\ldots,x_n$);\\[0ex]
~$n_m$\:=\:find\_repeat($m,x_1,x_2,\ldots,x_n$);\\[0ex]
~if($n_m\geq r$)\\[0ex]
~\hspace{2ex}for($i$\;=\;$1$;\;$i\leq r$;\;$i$\;=\;$i$\;+\;$1$)\\[0ex]
~\hspace{4ex}$v_{i}$\;=\;$x_{r+n_m-i+1}$\;-\;$m$;\\[0ex]
~else\;\Copen\\[0ex]
~\hspace{2ex}for($i$\;=\;$1$;\;$i\leq n_m$;\;$i$\;=\;$i$\;+\;$1$)\\[0ex]
~\hspace{4ex}$v_{i}$\;=\;$x_{r+n_m-i+1}$\;-\;$m$;\\[0ex]
~\hspace{2ex}for($i$\;=\;$n_m$\;+\;$1$;\;$i\leq r$;\;$i$\;=\;$i$\;+\;$1$)\\[0ex]
~\hspace{4ex}$v_{i}$\;=\;$x_{n+n_m-i+1}$\;-\;($m$\;-\;$1$);\Cclose\\[0ex]
}

The function {\code max($x_1,x_2,\ldots,x_n$)} simply returns the maximum value of the cells $x_1,x_2,\ldots,x_n$. The function {\code find\_repeat($m,x_1,x_2,\ldots,x_n$)} returns the number of times the value $m$ repeats in the cells $x_1,x_2,\ldots,x_n$. If the value of $n_m$ is at least $r$ then the buffer is stored between the $(n_m+1)$-st and $(n_m+r)$-th cells, and the buffer values are calculated by subtracting $m$ from the value of each cell. If the value of $n_m$ is less than $r$ then the buffer is stored cyclically in two cell groups: the last $r-n_m$ cells and the $n_m$ cells in locations $r+1,\ldots,r+{n_m}$. In the first group, the buffer values are given by subtracting $m-1$ from the cells' value and in the second group by subtracting $m$ from the cells' value.

\noindent
\textbf{Encoding map $\cE_{\textmd{buf}}$\kern1pt:}
The input to this map is a cell-state vector $\bfx = (x_1,x_2,\ldots,x_n)$, and a new bit $b$. Its output is either a cell-state vector $\bfy = (y_1,y_2,\ldots,y_n)$ or the erasure symbol $\E$.\hspace*{3ex}
\codemode{%
~($y_1,y_2,\ldots,y_n$)\;=\;($x_1,x_2,\ldots,x_n$); \\[0ex]
~$m$\:=\:max($x_1,x_2,\ldots,x_n$);\\[0ex]
~$n_m$\:=\:find\_repeat($m,x_1,x_2,\ldots,x_n$);\\[0ex]
~if($m$\;==\;$0$)\;\Copen\;   /\!/\;\textrm{\sl if this is the first write}       \\[0ex]
~\hspace{2ex}if($b$\;==\;$1$)\;$y_{r+1}$\;=\;$1$;\;\\[0ex]
~\hspace{2ex}else $y_{1}$\;=\;$1$;\;\Cclose\\[0ex]
~if($n_m$\;==\;$n$\;-\;$r$)\;\Copen\;   /\!/\;\textrm{\sl first write in this layer}\\[0ex]
~\hspace{2ex}for($i$\;=\;$1$;\;$i\leq n$\;-\;$r$\;+\;$1$;\;$i$\;=\;$i$\;+\;$1$)\\[0ex]
~\hspace{4ex}$y_i$\;=\;$m$;\;\\[0ex]
~\hspace{2ex}if($b$\;==\;$1$)\;$y_{r+1}$\;=\;$m$\;+\;$1$;\;\\[0ex]
~\hspace{2ex}else $y_{1}$\;=\;$m$\;+\;$1$;\;\Cclose\\[0ex]
~if($n_m < n$\;-\;$r$)\;\Copen\; /\!/\;\textrm{\sl not the first write in this layer} \\[0ex]
~\hspace{2ex}$y_{r+n_m+1}$\;=\;$y_{r+n_m+1}$\;+\;$b$;\\[0ex]
~\hspace{2ex}if($b$\;==\;$0$)\\[0ex]
~\hspace{4ex}for($i$\;=\;$1$;\;$i\leq n_m$\;+\;$r$;\;$i$\;=\;$i$\;+\;$1$)\\[0ex]
~\hspace{6ex}if($y_i$\;==\;$m$\;-\;$1$)\;\Copen\\[0ex]
~\hspace{8ex}$y_{r+n_m+1}$\;=\;$y_{r+n_m+1}$\;+\;$1$;\;break;\;\Cclose\;\Cclose\\[0ex]
~if($n_m \leq r$\;-\;$1$)\; /\!/\;\textrm{\sl one of first $r-1$ writes in this layer} \\[0ex]
~\hspace{4ex}$y_{n-r+1+n_m}$\;=\;$m$\;-\;$1$;\;\\[0ex]
}

On the first write, according to the bit value $b$, the first or the $(r+1)$-st cell changes its value to one. On the first write on each layer, the first $n-r+1$ cells are increased to level $m$, and then the first or the $(r+1)$-st cell is increased by one level, according to the bit value $b$. For all other writes, if the value if $b$ is one then we simply increase the $(r+n_m+1)$-st cell by one level, and otherwise we increase the first cell of level $m-1$ by one level. Finally, if it is one of the first $r-1$ writes in each level, then we need to update the last cell that stores the buffer to level $m-1$ since it no longer stores the buffer and thus its level has to be updated. 

Next, we prove the correctness of the construction.
\begin{lemma}\label{lem:max value}
After $s=x(n-r)+y$, where $1\leq y\leq n-r$, the maximum cell level is $x+1$ and there are $y$ cells in level $x+1$.
\end{lemma}
\begin{proof}
According to the encoding map $\cE_{\textmd{buf}}$, the maximum cell level increases every $n-r$ writes, on the $(i(n-r)+1)$-st write, for $0\leq i\leq q-2$. Therefore, after $s$ writes, the maximum cell value is $x=\left\lceil\frac{s}{n-r}\right\rceil$. 
If $y=1$ then the maximum cell value is $x+1$ and we can see that exactly one cell changes its value to $x+1$. For all other writes, the maximum cell value does not change and exactly one cell changes its value to the maximum cell value which is $x+1$.
\end{proof}
\begin{theorem}
The buffer code $\C(\cD_{\textmd{buf}},\cE_{\textmd{buf}})$ stores the buffer successfully and guarantees $t=(q-1)(n-r)$ writes.
\end{theorem}
\begin{proof}
According to Lemma~\ref{lem:max value}, after $t=(q-1)(n-r)$ writes the maximum cell level does not reach level $q$ and hence there is no need to erase the block of cells. We prove the correctness of the encoding and decoding maps to store the correct value of the buffer by induction on the number of writes $s$. This is done by proving that for all $1\leq s\leq t$, such that $s=x(n-r)+y$, where $1\leq y\leq n-r$, the buffer $(v_1,\ldots,v_r)$ is calculated successfully according to the decoding rules of the decoding map:
\begin{enumerate}
\item If $y\geq r$ then for $1\leq i\leq r$, $v_i = x_{r+y-i+1}-m$.
\item If $y<r$ then for $1\leq i\leq y$, $v_i = x_{r+y-i+1}-m$ and for $y+1\leq i\leq r$, $v_i = x_{r+y-i+1}-(m-1)$.
\end{enumerate}
It is straightforward to verify that after the first write the memory successfully stores the buffer. Assume the assertion is correct after the $s$-th write, where $1\leq s=x(n-r)+y\leq t-1$, $1\leq y\leq n-r$. Assume that the new bit to be written to the buffer on the $(s+1)$-st write is $b$ and let us consider the following cases:
\begin{enumerate}
\item If $y=n-r$, then on the $(s+1)$-st write in the encoding map the value of $n_m$ is $n-r$. Thus the first $n-r+1$ cells change their value to $m=x$, the values of the last $r-1$ cells do not change, and if $b=1$ then $y_{r+1}=m+1$, and otherwise $y_1=m+1$. Therefore, the new value of the buffer is also given according to the decoding rules.
\item If $y<n-r$, then $n_m=y<n-r$, and the value of the $(r+n_m+1)$-st cell increases by $b$ so the buffer is shifted one place to the right and it stores its updated value. If $b=0$, then we increase the first $n_m+1$ cells by one level. Note that $n_m=y$ and there are exactly $y$ cells with the maximum value so we can always find a cell of value less than $m$ and increase the value to $m$. Then, the buffer is again stored according to the above decoding rules.
\end{enumerate}
\end{proof}

\section{Conclusion}\label{sec:conclusion}
Rewriting codes for flash memories are important as they can increase the lifetime of the memory. Examples of such codes are flash codes~\cite{JBB07} and buffer codes~\cite{BJB07}. A significant contribution in this paper is an efficient construction of flash codes that support the storage of any number of bits. We show that the write deficiency order of the code is $O(k\log k \cdot \max\{\log_2 k,q\})$, which is an improvement upon the write deficiency order of the equivalent constructions in~\cite{JB08,JBB10,YVSW08}. The upper bound in~\cite{JBB07} on the guaranteed number of writes implies that the order of the lower bound on the deficiency is $O(kq)$. Therefore, there is a gap, which we believe can be reduced, between the write deficiency orders of our construction and the lower bound. For buffer codes, we showed how to improve an upper bound on the number of writes in the case where one cell is used to store the buffer. If there are multiple cells, we showed a construction that improves upon the one presented in~\cite{BJB07,JBB10}.

\end{document}